\documentclass[11pt]{article}

\usepackage{amssymb,amsmath,amsfonts}
\usepackage{graphicx,enumitem}
\usepackage{amsthm} % For proof environment
\usepackage{bm}
\usepackage[round]{natbib}
 
\usepackage{comment}

\usepackage{geometry}

\usepackage[dvipsnames]{xcolor}

\usepackage[colorlinks,citecolor=blue,urlcolor=blue]{hyperref}

\usepackage{tikz}
\tikzstyle{vertex}=[circle, draw, inner sep=2pt, fill=white]

\usepackage{pgf,pgfplots}
\usepackage{tikz-3dplot}
\usetikzlibrary{matrix,arrows,decorations.pathmorphing,shapes}
\usetikzlibrary{spy}
\usetikzlibrary{backgrounds}
\usetikzlibrary{decorations}
\usetikzlibrary{patterns}
\tikzset{snake it/.style={decorate, decoration=snake}}
\pgfdeclarelayer{background layer}
\pgfsetlayers{background layer,main}
\definecolor{darkgray}{rgb}{0.25,0.25,0.25}
\definecolor{darkgreen}{rgb}{0,0.5,0}
\definecolor{lightgray}{rgb}{0.75,0.75,0.75}
\colorlet{shadecolor}{red!5}

\renewcommand{\d}{{\mathrm{d}}}

\newcommand{\nablao}{\nabla^{\!\scriptscriptstyle\perp\!} }

\newcommand{\E}{\textsf{\upshape E}}

\renewcommand{\H}{{\mathbb H}}

\renewcommand{\P}{\textsf{\upshape P}}
\newcommand{\Qu}{\textsf{\upshape Q}}

\newcommand{\R}{{\mathbb R}}
\renewcommand{\S}{{\mathbb S}}
\newcommand{\N}{{\mathbb N}}

\newcommand{\Fcal}{{\mathcal F}}

\DeclareMathOperator{\rk}{rank}

\DeclareMathOperator*{\essinf}{ess\,inf}

\DeclareMathOperator{\tr}{tr}

\newtheorem{theorem}{Theorem}

\newtheorem{definition}[theorem]{Definition}

\newtheorem{proposition}[theorem]{Proposition}
\newtheorem{remark}[theorem]{Remark}

\numberwithin{equation}{section}
\numberwithin{theorem}{section}

\definecolor{red}{rgb}{1.00, 0.00, 0.00}

\definecolor{blue}{rgb}{0.00, 0.00, 1.00}

\definecolor{Darkgreen}{rgb}{0,0.7,0}

\definecolor{white}{rgb}{1,1,1}

\begin{document}

\title{Relative Arbitrage: Sharp Time Horizons and\\ Motion by Curvature\footnote{Acknowledgement: Part of this research was completed while we visited I.~Karatzas at Columbia University, whom we thank for his hospitality.  We also would like to thank A.~Banner, B.~Fernholz, M.~Shkolnikov, V.~Papathanakos, and especially I.~Karatzas and M.~Soner for helpful discussions. We are grateful to two anonymous referees for their careful reading and insightful comments.}}
\author{Martin Larsson\thanks{Department of Mathematical Sciences, Carnegie Mellon University, Wean Hall, 5000 Forbes Ave, Pittsburgh, Pennsylvania 15213, USA, \url{martinl@andrew.cmu.edu}.}
\and Johannes Ruf\thanks{Department of Mathematics, London School of Economics and Political Science, Columbia House, London, WC2A 2AE, UK, \url{j.ruf@lse.ac.uk}.}}

\maketitle

\begin{abstract}
We characterize the minimal time horizon over which any equity market with $d \geq 2$ stocks and sufficient intrinsic volatility admits relative arbitrage with respect to the market portfolio. If $d \in \{2,3\}$, the minimal time horizon can be computed explicitly, its value being zero if $d=2$ and $\sqrt{3}/(2\pi)$ if $d=3$. If $d \geq 4$, the minimal time horizon can be characterized via the arrival time function of a geometric flow of the unit simplex in $\R^d$ that we call the minimum curvature flow.

\bigskip


\end{abstract}

\section{Introduction}
Consider a market with $d \geq 2$ stocks, modeled by a $d$-dimensional vector $\mu = (\mu_1, \ldots, \mu_d)^\top$ of continuous nonnegative semimartingales summing to one. Each component of $\mu$ corresponds to the relative market weight of one stock; thus $\mu_i$ is the capitalization of the $i$-th stock divided by the total market capitalization.

\cite{FKVolStabMarkets} observed that if all market weights are strictly positive and satisfy
\begin{align} \label{eq:200212.1}
	\sum_{i = 1}^d \int_0^t \frac{1}{\mu_i} \d [\mu_i, \mu_i] \geq t \quad \text{for all $t \geq 0$},
\end{align}
then there exists some time $\overline T$ such that relative arbitrage is possible with respect to the market over any time horizon $[0,T]$ with $T > \overline T$ (precise definitions are provided below).  However, this left open the question whether such markets allow for relative arbitrage over arbitrary time horizons, including very short ones. About ten years later, \cite{Fernholz:Karatzas:Ruf:2016} showed that this is not always possible. They provided an example of a market $\mu$ that satisfies \eqref{eq:200212.1} but does not allow for short-term relative arbitrage. Still, it remains open how to characterize the precise time $T_* \in (0, \overline T]$ such that a market satisfying \eqref{eq:200212.1} always allows for relative arbitrage over $[0,T]$ when $T > T_*$, but not always when $T<T_*$.

In this paper we take a step toward answering this question. To do so, we work with a strengthening of \eqref{eq:200212.1}, only considering markets that satisfy the trace condition
\begin{align} \label{eq:200212.2}
	\tr[\mu,\mu](t)=\sum_{i = 1}^d [\mu_i, \mu_i](t)  \geq t \quad \text{for all } t \geq 0.
\end{align}
Our goal is to determine the smallest time $T_*$ such that every market satisfying \eqref{eq:200212.2} allows for relative arbitrage over any time horizon $[0,T]$ with $T > T_*$.

To achieve this goal, we formulate an optimal control problem. Its value function takes an argument $\mu_0$ in the unit simplex and returns a nonnegative number $T_*(\mu_0)$. This number is the infimum over those $T$ such that there exists a market $\mu$, that satisfies $\mu(0) = \mu_0$ and  \eqref{eq:200212.2}, and, moreover,  is free of relative arbitrage over $[0,T]$. Thanks to an application of the fundamental theorem of asset pricing, it actually suffices to consider markets $\mu$ that are martingales until they reach the boundary of the simplex.

In the case of $d=3$ stocks, the Hamilton--Jacobi--Bellman equation of this control problem (see \eqref{eq:200214.4} below) turns out to be a well studied partial differential equation arising in geometric analysis. Indeed, in this case the solution is the arrival time function of the so called mean curvature flow. Using known properties of this flow, we are able to explicitly compute the worst-case time $T_* = \max_{\mu_0} T_*(\mu_0)$, where the maximum is taken over the unit simplex in $\R^3$. Its value is $T_*=\sqrt{3}/(2\pi)$; see Theorem~\ref{T*}. In higher dimensions, the Hamilton--Jacobi--Bellman equation still has a geometric meaning, now as a description of what we call the minimum curvature flow. This flow is less well studied than the mean curvature flow, though it is closely related to the codimension flow of \cite{amb_son_96}. The partial differential equation for the arrival time of the minimum curvature flow is the subject of our companion paper \citet{LR:2020a}.

Here is the outline of this paper. Section~\ref{S:2} introduces some notation and the relevant financial concepts, such as trading strategies and relative arbitrage. Section~\ref{S:3} provides a representation of $T_* = \max_{\mu_0} T_*(\mu_0)$ in terms of an optimization problem that maximizes the essential infimum of  the exit time of  martingales with sufficient volatility. Section~\ref{S:4} discusses the case $d=3$, where an explicit numerical value for $T_*$ is obtained.  Next, Section~\ref{S:5} provides the results for the general case $d\geq 4$ by connecting the value function of the control problem to the arrival time of the so called minimum curvature flow.  Section~\ref{S:6} concludes by posing several open problems.

We end this introduction by providing some additional pointers to the literature discussing the existence of short-term relative arbitrage in the presence of sufficient intrinsic volatility, where intrinsic volatility has several but related meanings in the literature. \citet{Fe} observes that  there exists relative arbitrage over sufficiently long time horizons, provided that the market is diverse, i.e., $\max_{1\le i \le d}\mu_i \le 1-\delta$ for some $\delta > 0$ and the instantaneous covariance matrix of the stock returns is uniformly elliptic.  Under the same conditions, \citet{FKK} later prove that there exists relative arbitrage indeed over any time horizon. It is difficult to check statistically whether the instantaneous covariance matrix has eigenvalues bounded away from zero. Thus, \citet{FKVolStabMarkets} consider instead a scaled version of \eqref{eq:200212.1} which is easier to verify empirically. They prove that this suffices to guarantee relative arbitrage over long time horizons. \citet{Banner_Fernholz} prove that relative arbitrage over any time horizon exists if no stock defaults and the smallest stock is sufficiently volatile. \citet{FK} prove that in a uniformly elliptic Markovian model satisfying \eqref{eq:200212.1}, relative arbitrage over any time horizon exists. \citet{Pal:exponentially} shows the existence of asymptotic short-term arbitrage opportunities as the number of stocks tends to infinity and an appropriate notion of sufficient volatility is assumed to hold.
For some general introduction to stochastic portfolio theory and relative arbitrage we refer to \cite{FK_survey} and \cite{Vervuurt:2015}.

\section{Market models, arbitrage, and the smallest horizon} \label{S:2}
Let $\Delta^d = \{x\in[0,1]^d\colon x_1+\cdots+x_d=1\}$ for $d \geq 2$ denote the unit simplex in $\R^d$. A {\em market} or a {\em market weight process} is a $\Delta^d$-valued continuous semimartingale $\mu=(\mu_1,\ldots,\mu_d)^\top$ defined on some stochastic basis $(\Omega,\Fcal,(\Fcal_t)_{t\ge0},\P)$. In this paper, (in)equalities are understood $\P$-almost surely.  As mentioned in the introduction, each component of $\mu$ represents the relative capitalization of a stock. Note that $\mu$ is allowed to take values in all of $\Delta^d$, including the boundary.  

A {\em trading strategy} is an $\R^d$-valued $\mu$-integrable predictable process $\theta$, and its {\em relative value process} is defined by $V^\theta = \theta^\top \mu$. It is called {\em self-financing} if
\[
V^\theta = V^\theta(0) + \int_0^\cdot \theta^\top \d\mu.
\]
There is no bank account in this model. All self-financing trading strategies are fully invested in the stock market at all times. Moreover, wealth is measured in units of the total market capitalization, or \emph{relative to the market}. Put differently, the market portfolio is the numeraire, and the market weights $\mu_i$ are the stocks' market capitalizations expressed in this numeraire. This explains the definition of the relative value process, the self-financing condition, as well as the usage of the word \emph{relative}.

We now define the relevant arbitrage concept for this paper. This turns out to coincide with the classical no-arbitrage condition (NA), expressed in the numeraire currently in use. The terminology is chosen to emphasize this choice of numeraire.

\begin{definition} \label{D:200213.1}
Given a constant $T>0$, a trading strategy $\theta$ is called a {\em relative arbitrage over (time horizon) $[0,T]$} if $V^\theta \ge0$, $V^\theta(T)\ge V^\theta(0)$, and $\P(V^\theta(T)>V^\theta(0))>0$.
\end{definition}

Markets that satisfy \eqref{eq:200212.2} admit relative arbitrage over $[0,T]$, provided $T>1-|\mu(0)|^2$. This can be proved using (additively) functionally generated strategies, which are constructed as follows. Let $G\colon \R^d\to\R$ be a concave and $C^2$ function.
The functionally generated trading strategy associated with $G$ is then given by
\begin{equation*}
\theta = \nabla G(\mu) + \left(G(\mu) + \Gamma^G - \nabla G(\mu)^\top \mu\right) \bf{1}, 
\end{equation*}
where $\bf{1}$ denotes the $d$-dimensional column vector of ones, and
\begin{equation} \label{eq:GammaG}
\Gamma^G = - \frac12 \sum_{i,j =1}^d \int_0^\cdot \partial^2_{ij} G(\mu)  \d [\mu_i,\mu_j]
\end{equation}
is a nondecreasing finite-variation process.  An application of It\^o's formula yields that $\theta$ is indeed self-financing, with relative value process
\begin{equation*}
V^G = G(\mu) + \Gamma^G.
\end{equation*}
For further details on this class of trading strategies, we refer to \cite{Karatzas:Ruf:2016}, where additively functionally generated trading strategies are introduced and studied.

Consider now the function 
\begin{align} \label{eq:Q}
	Q(x) = 1 - |x|^2 = 1 -\sum_{i = 1}^d x_i^2, \qquad x\in \Delta^d.
\end{align}
We then get $\Gamma^Q = \tr [\mu, \mu]$ for the process in \eqref{eq:GammaG}, and the functionally generated trading strategy associated with $Q$ has relative value process $V^Q = Q(\mu) + \Gamma^Q\ge \Gamma^Q$. Since $\Gamma^Q(t) \geq t$ for all $t \geq 0$ if $\mu$ satisfies \eqref{eq:200212.2}, we obtain the following  proposition, well known in stochastic portfolio theory.
\begin{proposition} \label{P:200212}
	Any market that satisfies \eqref{eq:200212.2} allows for relative arbitrage over $[0,T]$, provided $T > V^Q(0) = 1-|\mu(0)|^2$. Moreover, this relative arbitrage can be implemented by the self-financing trading strategy 
	\[
	\theta = -2 \mu  + \left(1 + |\mu|^2  + \Gamma^Q\right) \bf{1}.
\]
\end{proposition}

The condition \eqref{eq:200212.2} is a statement about volatility, and is a crucial property for this paper. Giving it a name will help to make the statements below clear.

\begin{definition}\label{D_suff_vol}
A market  $\mu$ is called {\em sufficiently volatile} if \eqref{eq:200212.2} holds.
\end{definition}

If one used another concave function $G$ in place of $Q$, one would obtain another finite-variation process $\Gamma^G$. For example, \eqref{eq:200212.2} relies on $\Gamma^Q$, while \eqref{eq:200212.1} relies on the corresponding finite-variation process of the so called entropy function (see also Definition~\ref{D:200620} below and the discussion afterwards).

\begin{remark}
Usually in the literature one would find \eqref{eq:200212.2} (or, more specifically,  \eqref{eq:200212.1}) with the right-hand side multiplied by a strictly positive (but small) constant $\eta$. Here we assume $\eta=1$, which amounts to scaling time by a constant.
\end{remark}

How quickly can one obtain relative arbitrage in a sufficiently volatile market? In general, the answer depends on the market. We are interested in the worst case:
\[
\begin{minipage}[c][3em][c]{.7\textwidth}
\begin{center}
What is the smallest time horizon $T_*$ beyond which relative arbitrage is possible in any sufficiently volatile market?
\end{center}
\end{minipage}
\]
Without the qualification ``sufficiently volatile'' the answer is clearly $T_* = \infty$, since a constant market weight process rules out relative arbitrage over any time horizon. Slightly more formally, we write the time $T_*$ as follows:
 \begin{equation*}
T_* = \inf\left\{ T\ge0 \colon \ \ \begin{minipage}[c][2em][c]{.45\textwidth}
every sufficiently volatile market admits\\relative arbitrage over $[0,T]$  
\end{minipage}
\right\}.
\end{equation*}
Thus for $T>T_*$, any sufficiently volatile market admits relative arbitrage over $[0,T]$. For $T< T_*$, there exists at least one sufficiently volatile market that does not admit relative arbitrage over $[0,T]$.

The purpose of this paper is to characterize $T_*$.  If $d = 2$, it is known that $T_* = 0$. Indeed, \citet[Proposition~5.13]{Fernholz:Karatzas:Ruf:2016} show that every sufficiently volatile market with two assets allows for relative arbitrage over arbitrary time horizons. In the case $d \geq 3$, we know from \citet[Theorem~6.8]{Fernholz:Karatzas:Ruf:2016}  that
\[
	T_* \geq \max_{\mu_0 \in  \Delta^3} Q(\mu_0) - \frac{1}{2} = \frac{1}{6},
\]
where $Q$ is the quadratic function introduced in \eqref{eq:Q};
see also Subsection~\ref{SS:4.2} below. Moreover, it is clear from Proposition~\ref{P:200212} that 
\[
	T_* \leq \max_{\mu_0 \in \Delta^d} Q(\mu_0) = 1 - \frac{1}{d}.
\]
In the following sections we will prove that $T_* = \sqrt{3}/(2\pi)$, showing that both bounds are not tight.

\section{A representation of $T_*$ in terms of martingales} \label{S:3}
In this section we use a version of the fundamental theorem of asset pricing to prove that it suffices to only consider martingales when computing $T_*$. To this end, for any $\R^d$-valued process $\nu$, define
\[
\zeta_\nu=\inf \left\{t\ge0\colon \nu(t) \notin \Delta^d \right\}.
\]
Recall that the \emph{essential infimum} of a random variable $X$, denoted $\essinf X$, is the largest deterministic lower bound on $X$. That is, $\essinf X=\sup\{c\in\R\colon\P(X\ge c)=1\}$.

\begin{theorem}\label{T_Tstar_repr}
We have the representation
\begin{equation} \label{eq:T* essinf}
T_* = \sup\left\{ \essinf \zeta_\nu\colon \ \ \begin{minipage}[c][2em][c]{.45\textwidth}
$\nu$ is an $\R^d$-valued  continuous martingale\\ with $\tr[\nu, \nu](t) \geq t$ for all $t \geq 0$
\end{minipage}
\right\}.
\end{equation}
\end{theorem}

\begin{proof}
Denote the right-hand side of \eqref{eq:T* essinf} by $T_{**}$. For any $T<T_{**}$ there exists a continuous martingale $\nu$ with $\tr[\nu, \nu](t) \geq t$ for all $t \geq 0$ that remains in $\Delta^d$ on $[0,T]$. Using this martingale, one can easily construct a sufficiently volatile market $\mu$ that is a martingale on $[0,T]$. The martingale property implies that no relative arbitrage on $[0,T]$ can be constructed. This yields $T\le T_*$ and we deduce that $T_{**}\le T_*$.
 
 Suppose now for contradiction that $T_{**}< T_*$, and choose $T\in(T_{**},T_*)$.   Then there exists a sufficiently volatile market  $\mu$, defined on a stochastic basis $(\Omega,\Fcal,(\Fcal_t)_{t\ge0},\P)$, that does not admit relative arbitrage over $[0,T]$. We will show below that there exists a probability measure $\Qu\ll\P$ such that $\mu$ is a $\Qu$-martingale on $[0,T]$. Then it is easy to construct an $\R^d$-valued $\Qu$-martingale $\nu$, possibly on an extension of the stochastic basis $(\Omega,\Fcal,(\Fcal_t)_{t\ge0},\Qu)$, such that $\nu = \mu$ on $[0,T]$ and  $\tr [\nu, \nu](t) \geq t$ for all $t \geq 0$. In particular, since $\mu$ takes values in $\Delta^d$, we have $\zeta_\nu \geq T$. This means that $T \leq T_{**}$, a contradiction, showing that $T_{**} = T_*$ as claimed.
 
 It remains to argue the existence of the probability measure $\Qu$. Thanks to \citet[Theorem~1.4]{DS:95}, we only need to show that $\mu$ satisfies the no-arbitrage condition (NA) on $[0,T]$. Namely, that there exists no trading strategy $\theta$ such that $\theta^\top(0) \mu(0) \geq 0$, $\int_0^{\cdot \wedge T} \theta^\top \d \mu \geq -\kappa$ for some $\kappa \geq 0$, $\int_0^T \theta^\top \d \mu \geq 0$, and $\P(\int_0^T \theta^\top \d \mu>0)>0$. Assume for contradiction that such $\theta$ exists. Define the predictable process
 \[
 	\overline \theta = \theta + \left(\kappa + \int_0^\cdot \theta^\top \d \mu - \theta^\top \mu \right) \bf{1}.
 \]
Since  $\int_0^\cdot \overline \theta^\top \d \mu = \int_0^\cdot  \theta^\top \d \mu$, one sees that $\overline \theta$ is self-financing and a relative arbitrage in the sense of Definition~\ref{D:200213.1}. This contradicts the fact that $\mu$ does not admit relative arbitrage over $[0,T]$. We deduce that $\mu$ satisfies the no-arbitrage condition (NA) on $[0,T]$.
 \end{proof}

On the right-hand side of \eqref{eq:T* essinf}, the inequality in $\tr[\nu, \nu](t) \geq t$ can be replaced by an equality without changing the value of the supremum. To see this, consider $\tilde\nu(t)=\nu(A(t))$, where $A$ is the right-continuous inverse of $\tr[\nu, \nu]$. Then $\tilde\nu$ is an $\R^d$-valued continuous martingale with $\tr[\tilde\nu, \tilde\nu](t) = t$ for all $t \geq 0$, and $\zeta_{\tilde\nu}\ge\zeta_\nu$. Motivated by this, we define the function $u\colon \R^d \rightarrow  [0, \infty)$ by
	\begin{equation} \label{eq:u}
u(y) = \sup\left\{ \essinf \zeta_\nu\colon \ \ \begin{minipage}[c][2em][c]{.55\textwidth}
$\nu$ is an $\R^d$-valued  continuous martingale with\\  $\nu(0) = y$ and $\tr[\nu, \nu](t) = t$ for all $t \geq 0$
\end{minipage}
\right\}.
\end{equation}
Thanks to Theorem~\ref{eq:T* essinf} we then have
\[
T_* = \sup_{y \in \R^d} u(y) = \sup_{y \in  \Delta^d } u(y).
\]

 \section{The case $d = 3$ and the appearance of mean curvature flow} \label{S:4}
  
In this section we focus on the case $d = 3$.   We proceed in several steps.  First, in Subsection~\ref{SS:4.1} we map the hyperplane containing $ \Delta^3$ to $\R^2$. Then, in Subsection~\ref{SS:4.2} we recall how \cite{Fernholz:Karatzas:Ruf:2016} obtained a lower (but not sharp) bound on $T_*$. Motivated by their approach, in Subsection~\ref{SS:4.3} we introduce and discuss a boundary value problem whose solution will be used to characterize $T_*$. Subsection~\ref{SS:4.4} discusses the existence of a weak solution to a related stochastic differential equation. Finally, Subsection~\ref{SS:4.5} provides a computation of $T_*$.

\subsection{Mapping the market $\mu$ to a two-dimensional process}  \label{SS:4.1}

A market $\mu$ with $d=3$ assets evolves in the unit simplex $\Delta^3$, which is a two-dimensional subset of $\R^3$. Using a suitable projection $U\colon\R^3\to\R^2$, one can express $\mu$ in terms of a two-dimensional process $X$. This is illustrated in Figure~\ref{F:1}. Formally, let $U\in\R^{2\times 3}$ be a matrix with orthonormal rows and $U\bm1=0$. Equivalently, the rows of $U$ form an orthonormal basis for $\H -  \frac13 \bm1$, where
\[
	\H= \left\{x\in\R^3\colon x^\top\bm1=0\right\} + \frac{1}{3} \bm1
\]
is the two-dimensional plane containing $\Delta^3$.  For example, we may choose
\[
	U = \left(
	\begin{matrix}
		\frac{1}{\sqrt{2}} & \frac{-1}{\sqrt{2}} & 0\\
		\frac{1}{\sqrt{6}} & \frac{1}{\sqrt{6}} & \frac{-2}{\sqrt{6}}
	\end{matrix}
	\right).
\]
The map $y\mapsto Uy$ from $\H$ to $\R^2$ is a bijection with inverse $x\mapsto U^\top x+\frac13\bm1$. Since the rows of $U$ are orthonormal, this map is an isometry. Thus $\nu$ is an $\H$-valued martingale if and only if $X=U\nu$ is an $\R^2$-valued martingale, and in this case $U^\top[X,X]U=[\nu,\nu]$. In particular, the quadratic variations satisfy $\tr[\nu,\nu]=\tr[X,X]$. The simplex $\Delta^3$ is mapped to the compact domain
\[
K = U(\Delta^3) \subset \R^2.
\]
An advantage of this transformation is that $K$ has an interior (in $\R^2$), while $\Delta^3$ does not (in $\R^3$).

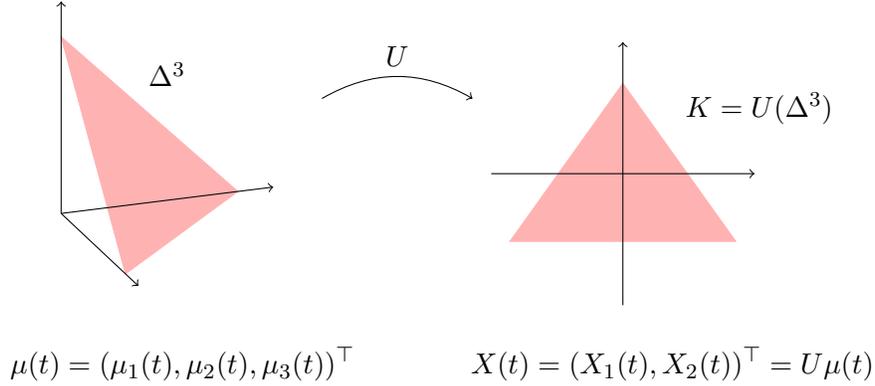
\begin{figure}[h]\begin{center}
\tdplotsetmaincoords{70}{70}
\begin{tikzpicture}[tdplot_main_coords]
  \def\laxis{3}
  \def\ltriangle{2.5}
  \filldraw [red!30] (\ltriangle,0,0) -- (0,\ltriangle,0) -- (0,0,\ltriangle) -- cycle;
  \begin{scope}[->]
    \draw (0,0,0) -- (\laxis,0,0) node [below] {};
    \draw (0,0,0) -- (0,\laxis,0) node [above] {};
    \draw (0,0,0) -- (0,0,\laxis) node [left] {};
  \end{scope}
  \begin{scope}[]
    \draw (\laxis,0,\laxis) node [right] {${\Delta^3}$};
  \end{scope}
\end{tikzpicture}
\tdplotsetmaincoords{0}{0}
\begin{tikzpicture}[tdplot_main_coords]
\path[->] (-4,1) edge [bend left] node[above] {$U$} (-2,1);
  \def\laxis{3.5}
  \def\lenD{3}
  \filldraw [red!30] (-0.5*\lenD,-0.3*\lenD,0) -- (0.5*\lenD,-0.3*\lenD,0) -- (0,0.4*\lenD,0) -- cycle;
  \begin{scope}[->]
    \draw (-0.5*\laxis,0,0) -- (0.5*\laxis,0,0) node [below] {};
    \draw (0,-0.5*\laxis,0) -- (0,0.5*\laxis,0) node [right] {};
  \end{scope}
  \begin{scope}[]
    \draw (0.2*\laxis,0.25*\laxis,0) node [right] {$K=U({\Delta^3})$};
  \end{scope}
\end{tikzpicture}
\end{center}
\[
\mu(t)=(\mu_1(t),\mu_2(t), \mu_3(t))^\top \qquad\qquad X(t)=(X_1(t), X_2(t))^\top = U\mu (t)
\]
\caption{The mapping from the unit simplex $\Delta^3$ to $K$.}
\label{F:1}
\end{figure}

For any $\R^2$-valued semimartingale $X$ we define
\begin{align} \label{200214.3}
  	\tau_X=\inf \left\{t\ge0\colon X(t) \notin K \right\}.
\end{align}
We also define the value function $v\colon \R^2 \rightarrow  [0, \infty)$ by
	\begin{equation} \label{eq:v}
v(x) = \sup\left\{ \essinf \tau_X\colon \ \ \begin{minipage}[c][2em][c]{.53\textwidth}
$X$ is an $\R^2$-valued  continuous martingale with\\  $X(0) = x$ and $\tr[X, X](t)= t$ for all $t \geq 0$
\end{minipage}
\right\}.
\end{equation}
The isometry property of the map $y\mapsto Uy$ implies that the function $v$ in \eqref{eq:v} is related to $u$ in \eqref{eq:u} by the identity
\[
	u(y) = v(Uy), \qquad y \in \Delta^3.
\]	
As a consequence, we also get
\[
	T_* = \sup_{x \in \R^2} v(x) = \sup_{x \in  K } v(x).
\]

\subsection{The lower bound of \cite{Fernholz:Karatzas:Ruf:2016}} \label{SS:4.2}
The above representation of $T_*$ indicates that we should look for $\R^2$-valued martingales that do not slow down (in the sense that the trace of the quadratic variation remains equal to $t$), yet remain in $K$ for a deterministic amount of time.
To make headway, let us  recall how  \cite{Fernholz:Karatzas:Ruf:2016} derive the lower bound $T_* \geq 1/6$, mentioned at the end of Section~\ref{S:2}. They consider the stochastic differential equation
\begin{align} \label{200214.1}
\d X(t) = \d\begin{pmatrix} X_1(t)\\X_2(t)\end{pmatrix} = \frac{1}{\sqrt{X^2_1(t)+X^2_2(t)}} \begin{pmatrix} X_2(t)\\-X_1(t)\end{pmatrix} \d W(t), \qquad t \geq 0,
\end{align}
where $W$ denotes a one-dimensional Brownian motion.  It can be shown that this stochastic differential equation always admits a weak solution $X$ satisfying $\tr[X, X](t) = t$ for all $t \geq 0$, even when $X(0) = 0$; see Theorem~\ref{T:200215} and Remark~\ref{R:Pittsburgh2} below.
An application of It\^o's formula yields
\begin{align} \label{200214.2}
	|X(t)|^2 = X_1^2(t) + X_2^2(t) = t, \qquad t \geq 0.
\end{align}
Thus at any time $t$, $X(t)$ lies on a centered circle with radius $\sqrt{t}$; see Figure~\ref{F:2} for a simulated sample path of $X$. It is clear that with $\tau_X$ as in \eqref{200214.3} we have $\essinf\tau_X = r^2$, where $r$ is the radius of the largest circle contained in $K$. Basic Euclidean geometry yields $r=1/\sqrt{6}$, and hence $\essinf\tau_X=1/6$. This gives the bound $T_* \geq 1/6$.

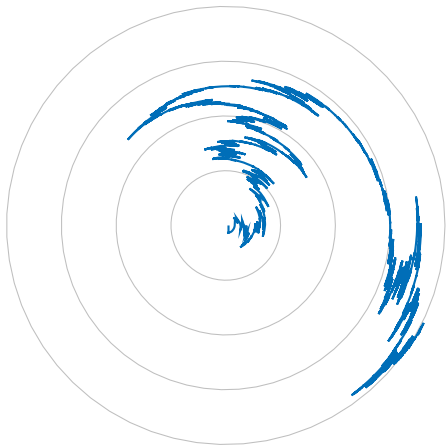
\begin{figure}[h]
\begin{center}
\begin{tikzpicture}
  \begin{axis}[
    width=0.5\textwidth,
    height=0.5\textwidth,
    axis line style={draw=none},
    tick style={draw=none},
    xticklabels = {,,},
    yticklabels = {,,},
    xmin = -4.1, xmax = 4.1,
    ymin = -4.1, ymax = 4.1,
  ]
  \addplot [gray!50, thin, domain=0:2*pi,samples=60]({1*cos(deg(x))},{1*sin(deg(x))});
  \addplot [gray!50, thin, domain=0:2*pi,samples=60]({2*cos(deg(x))},{2*sin(deg(x))});
  \addplot [gray!50, thin, domain=0:2*pi,samples=60]({3*cos(deg(x))},{3*sin(deg(x))});
  \addplot [gray!50, thin, domain=0:2*pi,samples=60]({4*cos(deg(x))},{4*sin(deg(x))});
  \addplot[no markers, color=NavyBlue, thick] file[skip first] {BM_on_circle_1.txt};
\end{axis}
\end{tikzpicture}
\caption{A sample path of the solution of \eqref{200214.1}.}
\label{F:2}
\end{center}
\end{figure}

Before proceeding, we observe that \eqref{200214.1} can be written
 \begin{align} \label{200214.1'}
\d X(t) = \frac{\nablao \overline Q(X(t))}{|\nabla \overline Q(X(t))|} \d W(t), \qquad t \geq 0,
\end{align}
where $\overline Q: \R^2 \ni x \mapsto 1-x_1^2 - x_2^2$, and $\nablao \overline Q = \begin{pmatrix} -\partial_2\overline Q\\ \partial_1 \overline Q\end{pmatrix}$ is the so-called skew-gradient of $\overline Q$.

\subsection{The mean curvature equation} \label{SS:4.3}

The concentric circles in Figure~\ref{F:2} are poorly adapted to the triangular geometry of $K$. For this reason $1/6$ is not a sharp lower bound for $T_*$. The goal is now to replace $Q$ in \eqref{200214.1'} by some other function that better reflects the geometry of $K$. To illustrate the idea, consider a function $w \in C^2$ and assume, for the moment, that the stochastic differential equation
   \begin{align} \label{200215.2}
\d X(t)  = \frac{\nablao w(X(t))}{|\nabla w(X(t))|} \d W(t), \qquad t \geq 0, 
\end{align}
has a weak solution starting from any point in $K$. Using It\^o's formula we then obtain
\begin{align} \label{200214.4}
t + w(X(t)) = w(X(0)) + \int_0^t{ \left(1 + \frac{\nablao w^\top \nabla^2 w \nablao w}{2|\nabla w|^2}(X(s))\right)} \d s, \qquad t \geq 0.
\end{align}
Let us additionally assume that the function $w$ can be chosen so that the integrand is zero and $w|_{\partial K} = 0$.  That is, we assume $w$ is a solution of the boundary value problem
\begin{align} \label{eq:200214.4}
\left\{
\begin{aligned}
1 + \frac{\nablao w^\top \nabla^2 w \nablao w}{2|\nabla w|^2} &=0 && \text{in $K^\circ$} \\
w &=0 && \text{on $\partial K$.}
\end{aligned}
\right.
\end{align}
With all these assumptions in place we set $t=\tau_X$ in \eqref{200214.4} to get $\tau_X = w(X(0))$.   This yields $v(X(0)) \geq\essinf\tau_X= w(X(0))$, where $v$ is the value function defined in \eqref{eq:v}. Hence we get the lower bound $T_* \geq \sup_{x \in K} w(x)$.

To turn these ideas into a rigorous argument, two main issues need to be resolved: solving the boundary value problem \eqref{eq:200214.4} and finding a solution of the stochastic differential equation \eqref{200215.2}. We deal with the latter in the next subsection. Here, we focus on \eqref{eq:200214.4}.

It turns out that \eqref{eq:200214.4} describes the arrival time of the so-called \emph{mean curvature flow of $\partial K$}, and we now discuss the physical phenomenon that this represents. The mean curvature (or curve shortening) flow is a construction that gradually deforms a given initial contour, in our case $\partial K$. Each point on the contour moves in the normal direction at a speed equal to the curvature at that point (in our case, half the curvature). Figure~\ref{F:3} illustrates this. The \emph{arrival time function} $w$ maps each point $x$ to the time $w(x)$ it takes for the evolving contour to reach $x$. If the initial contour is convex (i.e., encloses a convex region), then $w(x)$ is well-defined and finite for every point in the enclosed region. The contour at a positive time $t>0$ is then $\{x 
\in K^\circ \colon w(x)=t\}$. The contour eventually shrinks to a point and disappears. This happens at a finite extinction time $t^*<\infty$.

\begin{figure}[h]
\begin{center}
\begin{tikzpicture}
  \coordinate (A) at (0,0);
  \coordinate (B) at (1,1);
  \coordinate (C) at (2.5,1);
  \coordinate (D) at (4,-1);
  
  \draw [xshift=-5, yshift=2.5, color=BrickRed, thick, ->] plot coordinates {(2.5,1) (2.5-0.09,1-0.25)};
  \draw [xshift=1, yshift=1, color=BrickRed, thick, ->] plot coordinates {(4,-1) (4-0.8,-1+0.5)};
  \draw [xshift=-2.5, yshift=3, color=BrickRed, thick, ->] plot coordinates {(0,0) (0.5,0.17)};
  \draw [xshift=-2.7, yshift=-3.3, color=BrickRed, thick, ->] plot coordinates {(2,-0.5) (2+0.04,-0.5+0.15)};

  \draw [color=NavyBlue, thick] plot [smooth cycle] coordinates {(A) (B) (C) (D)};
\end{tikzpicture}
\caption{An illustration of the mean curvature flow in $\R^2$. The parts of the contour with high curvature move inwards fast, indicated by long arrows. Flat pieces move slowly, indicated by short arrows. The flow continues until the extinction time $t^*$.} \label{F:3}
\end{center}
\end{figure}
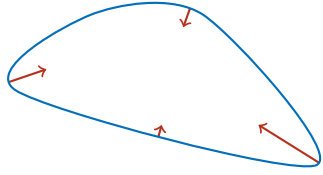

In our case, the initial contour $\partial K$ is not smooth, and some care is needed to define the flow. The contour is however convex, and the mean curvature flow can be shown to exist and have a well-defined arrival time function. Its level curves are drawn schematically in Figure~\ref{F:4}. The arrival time function $w$ corresponding to $\partial K$ turns out to be the solution of \eqref{eq:200214.4}. This becomes clearer when the equation by simple algebra is written in the equivalent form
\begin{align} \label{eq:200215.1}
\frac{1}{|\nabla w|} = - \frac12\, {\rm div}\left( \frac{\nabla w}{|\nabla w|} \right),
\end{align}
where ${\rm div}(f)=\partial_1f_1+\partial_2f_2$ denotes the divergence of a function $f=(f_1,f_2)\colon\R^2\to\R^2$. Evaluated at a point $\bar x$ on a contour $\{x \in K\colon w(x)=t\}$, the right hand side of \eqref{eq:200215.1} gives half the curvature of the contour at that point. The size of the gradient, $|\nabla w(\bar x)|$, is expressed in units of time per distance, and gives the change in arrival time per unit displacement orthogonally to the level curve at $\bar x$. Thus $1/|\nabla w(\bar x)|$, the left-hand side of \eqref{eq:200215.1}, is the speed by which the contour moves inwards.

\begin{figure}[h]
\begin{center}
\begin{tikzpicture}
  \def\sx{2}
  \def\sy{1.7}
  \coordinate (A) at ($\sx*(-1,0)+\sy*sqrt(3)/3*(0,-1)$);
  \coordinate (B) at ($\sx*(1,0)+\sy*sqrt(3)/3*(0,-1)$);
  \coordinate (C) at ($\sy*2*sqrt(3)/3*(0,1)$);
  \draw [color=NavyBlue, thick] plot coordinates {(A) (B) (C) (A)};

  \def\sxx{\sx/1.5}
  \def\syy{\sy/1.5}
  \coordinate (A) at ($\sxx*(-1,0)+\syy*sqrt(3)/3*(0,-1)$);
  \coordinate (B) at ($\sxx*(1,0)+\syy*sqrt(3)/3*(0,-1)$);
  \coordinate (C) at ($\syy*2*sqrt(3)/3*(0,1)$);
  \draw [color=NavyBlue, thick] plot [smooth cycle, tension=0.5] coordinates {(A) (B) (C)};

  \def\sxx{\sx/3}
  \def\syy{\sy/3}
  \coordinate (A) at ($\sxx*(-1,0)+\syy*sqrt(3)/3*(0,-1)$);
  \coordinate (B) at ($\sxx*(1,0)+\syy*sqrt(3)/3*(0,-1)$);
  \coordinate (C) at ($\syy*2*sqrt(3)/3*(0,1)$);
  \draw [color=NavyBlue, thick] plot [smooth cycle, tension=1.1] coordinates {(A) (B) (C)};

  \def\sxx{\sx/9}
  \def\syy{\sy/9}
  \coordinate (A) at ($\sxx*(-1,0)+\syy*sqrt(3)/3*(0,-1)$);
  \coordinate (B) at ($\sxx*(1,0)+\syy*sqrt(3)/3*(0,-1)$);
  \coordinate (C) at ($\syy*2*sqrt(3)/3*(0,1)$);
  \draw [color=NavyBlue, thick] plot [smooth cycle, tension=1.5] coordinates {(A) (B) (C)};

  \draw node at (1,1) [above] {$\partial K$};
\end{tikzpicture}
\caption{Mean curvature flow with initial contour $\partial K$. The flow deforms $\partial K$ to resemble small circles (see \citet{Gage:1984}), before disappearing at the center of the triangle $K$. } \label{F:4}
\end{center}
\end{figure}
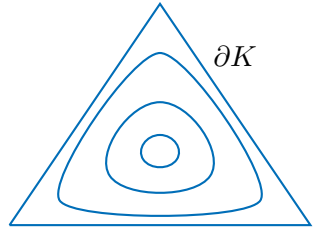

We are interested in classical solutions of \eqref{eq:200214.4}. Some care is needed, because the left-hand side of \eqref{eq:200214.4} is undefined at critical points of $w$ (i.e., at points where the gradient vanishes). With this in mind, we call a function $w\in C(K)\cap C^2(K^\circ)$ a \emph{classical solution} of \eqref{eq:200214.4} if it satisfies the equation at every point $x\in K$ where $\nabla w(x)\ne0$, and if it satisfies the boundary condition $w|_{\partial K}=0$. Here one immediately runs into problems of uniqueness: under the above definition, the zero solution $w=0$ is a classical solution, simply because \emph{every} point in $K^\circ$ is a critical point. Therefore, we will not have uniqueness among \emph{all} classical solutions, but only among those that have a finite number of critical points. 
Note that this definition of classical solution is stated slightly differently from the one  in \citet[Section~4]{col_min_16}.

\begin{theorem} \label{T:200217.1}
The boundary value problem \eqref{eq:200214.4} has a unique classical solution $w$ with a single critical point $\bar x\in K^\circ$. The solution $w$ is strictly positive, continuous on $K$, and $C^3$ in $K^\circ$, and the Hessian at the critical point is $\nabla^2 w(\bar x)=-2I$.
\end{theorem}

\begin{proof}
\citet[Section~5]{eva_spr_91} (see also \citet{che_gig_got_91}) establish the existence of a so-called \emph{generalized mean curvature flow} $(\Gamma_t)_{t \geq 0}$, where $\Gamma_0=\partial K$ is the boundary of the convex domain $K$ and each $\Gamma_t$ is a closed curve. Furthermore, by \citet[Theorem~5.5]{eva_spr_92}, there exists some $t^*>0$ such that $\Gamma_t$ is in fact a smooth convex curve for all $t\in(0,t^*)$ and empty for all $t > t^*$. 
\citet[Section~6]{eva_spr_91} then assert that $(\Gamma_t)_{0<t<t^*}$ evolves by mean curvature in the classical differential geometric sense  up to the extinction time $t^*$ (see e.g., \cite{GH:86}, in particular Theorem~1.1). 
Thanks to the remark on page~75 in \cite{GH:86}, $\Gamma_t$ is also strictly convex, i.e., does not contain any line segments, for all $t \in (0,t^*)$. 

We next argue that 
\begin{equation}\label{eq_coffeeshop_2}
\bigcup_{t > 0} \Gamma_t = K^\circ.
\end{equation}
First of all, we  have $\bigcup_{t > 0} \Gamma_t \subset K^\circ$ since $(\Gamma_t)_{0<t<t^*}$ evolves by mean curvature in the classical differential geometric sense.
Assume for contradiction that the other set inclusion does not hold, and let $\overline \Gamma_0$ denote the boundary of $\bigcup_{t > 0} \Gamma_t$.  Let $\overline t^*$ denote the extinction time of the generalized mean curvature flow $(\overline \Gamma_t)_{t \geq 0}$ with initial contour $\overline \Gamma_0$. Thanks to the comparison principle of \citet[Theorem~7.2]{eva_spr_91} we have $\overline t^* = t^*$. 
Let now $\delta > 0$ denote the strictly positive distance between $\Gamma_0$ and $\overline \Gamma_0$ and choose $\varepsilon \in (0, t^*)$ small enough so that the diameter of $\Gamma_{t^*-\varepsilon}$ is strictly less than $\delta$. Then the distance between $\Gamma_{t^*-\varepsilon}$ and $\overline \Gamma_{t^*-\varepsilon}$ is strictly less than $\delta$, leading to a contradiction with \citet[Theorem~7.3]{eva_spr_91}.  This proves \eqref{eq_coffeeshop_2}.

Let $\overline w: K \rightarrow [0, t^*]$ denote the arrival time function: for each $x\in K$, $\overline w(x)=t$ if $x\in\Gamma_t$.  This is well defined and positive on $\bigcup_{t > 0} \Gamma_t = K^\circ$ thanks to the arguments in \citet[Section~7]{eva_spr_91}, and zero on $\partial K$.  We claim that $\overline w$ is a $C^3$ classical solution of the partial differential equation
\begin{align} \label{200223.1}
1 + \frac{\nablao \overline w^\top \nabla^2 \overline w \nablao \overline w}{|\nabla \overline w|^2} =0
\end{align}
in $K^\circ$ with a single critical point, at which the Hessian is $-I$. To show this, fix any $t_0\in(0,t^*)$, and let $\widetilde w(x)= \overline w(x)-t_0$ be the arrival time of the time-shifted flow $\widetilde\Gamma_t=\Gamma_{t_0+t}$. The initial contour of this flow is the smooth and strictly convex curve $\Gamma_{t_0}$. By \citet[Lemma~3.1]{koh_ser_06} (see also \cite{Huisken:1993} for an earlier $C^2$ regularity result), $\widetilde w$ is a $C^3$ classical solution of \eqref{200223.1} in the domain enclosed by $\Gamma_{t_0}$, has a single critical point $\bar x$ in this domain, and satisfies $\nabla^2\widetilde w(\bar x)=-I$. Since $t_0$ was arbitrary, the properties carry over to $\overline w$ in the whole domain $K^\circ$. It is clear that the critical point $\bar x$ does not depend on $t_0$.

Next, we argue that $\overline w$ is continuous up to the boundary. Indeed, let $(x_n)_{n \in \N} \subset  K^\circ$ converge to some $x\in\partial K$ and suppose for contradiction that $t_n=\overline w(x_n)$ converges to a strictly positive limit. Then $\bar t=\inf_{n \in \N} t_n>0$, and $x_n\in\bigcup_{t\ge\bar t}\Gamma_t$ for all $n \in \N$. However, then  $(x_n)_{n \in \N}$ remains bounded away from $\partial K$, a contradiction.

Finally, defining $w=2\overline w$, we obtain a classical solution of \eqref{eq:200214.4} with the required properties.

It remains to argue uniqueness of a classical solution  for \eqref{eq:200214.4} with one critical point. Thanks to the comparison principle for viscosity solutions in \citet[Theorem~4]{koh_ser_06} it suffices to argue that each classical solution for \eqref{eq:200214.4} with one critical point is also a viscosity solution in the sense of \citet[Definition~3]{koh_ser_06}.  We will not give the details of this rather standard argument, but only indicate some key points. Viscosity solutions  for \eqref{eq:200214.4} can be defined using the symbol
\[
F(p,M) = - \frac{M_{11}p_2^2 - 2p_1p_2M_{12} + M_{22}p_2^2}{2|p|^2}, \qquad p \in \R^2 \setminus \{0\}, M \in \S^2,
\]
where $\S^2$ denotes the set of $2\times 2$ symmetric matrices. If $u$ is a classical solution with a single critical point $\bar x$, then \eqref{eq:200214.4} states that $F(\nabla u(x),\nabla^2 u(x))=1$ for all $x \in K \setminus \{\bar x\}$. Letting $F^*$ and $F_*$ denote the upper and lower semicontinuous envelopes of $F$, it follows that $F_*(\nabla u(\bar x),\nabla^2 u(\bar x))\le1\le F^*(\nabla u(\bar x),\nabla^2 u(\bar x))$. Using (degenerate) ellipticity of $F$, it is now straightforward to verify that $u$ is a viscosity solution of \eqref{eq:200214.4}, as required. This completes the proof of the theorem.
\end{proof}

\begin{remark} \label{R:Pittsburgh1}
Although we have a triangle shaped domain $K$ in mind, Theorem~\ref{T:200217.1} holds true for any compact convex subset of $\R^2$.   If $K$ is indeed the triangle, then by symmetry the critical point $\bar x$ must be the center point of the triangle. Specifically, with $K$ as in Subsection~\ref{SS:4.1}, we have $\bar x = 0$.
\end{remark}

Let us  mention that the mean curvature flow is more commonly studied via a parabolic equation satisfied by $v(t,x)=w(x)-t$; see for instance \cite{eva_spr_91} and \cite{che_gig_got_91}. The elliptic equation in \eqref{eq:200214.4} for the arrival time is however more natural in our context. This equation was first studied by \citet[Subsection~7.3]{eva_spr_91}, and subsequently by a number of other authors. In particular, \cite{koh_ser_06} discuss a deterministic game interpretation of \eqref{eq:200214.4} that is related to the stochastic representation that we use here. Other  stochastic representations of the mean curvature flow have been obtained by \cite{BCQ:01} and \cite{son_tou_02,son_tou_02b,son_tou_03}. These 
closely related methods have been very useful in solving a variety of stochastic control problems.

\subsection{Solving the stochastic differential equation in \eqref{200215.2}} \label{SS:4.4}

We assumed above that the stochastic differential equation in \eqref{200215.2} allows for a weak solution whenever $X(0) \in K^\circ$. We now argue that that such a weak solution exists, at least up to a deterministic time.

\begin{theorem} \label{T:200215}
Let $w$ be the classical solution of the boundary value problem~\eqref{eq:200214.4} from Theorem~\ref{T:200217.1} and fix $x_0\in K$. Then there exists an $\R^2$-valued martingale $X$ with $X(0) = x_0$  and  $\tr [X,X](t) = t$ for all $t \geq 0$ that satisfies the stochastic differential equation in \eqref{200215.2}  on $[0, \tau_X]$.  
Furthermore, $w(X(t))=w(x_0)-t$ for all $t \in[0, w(x_0)]$, and $\tau_X = w(x_0)$.
\end{theorem}
\begin{proof}
The  proof is similar to the proof of \citet[Lemma~3.2]{LR:2020a}.
If $x_0 \in \partial K$  we  let $X$ be an $\R^2$-valued Brownian motion starting in $x_0$, scaled by $1/\sqrt{2}$ to ensure that $\tr [X,X](t) = t$ for all $t \geq 0$.  

Suppose that $x_0 \in K^\circ \setminus \{\bar x\}$, where $\bar x$ is the unique critical point and global maximum of $w$. Since the stochastic differential equation has locally Lipschitz coefficients on the set 
\[
	\{x \in K^\circ\colon \nabla w(x)\ne0\}=K^\circ\setminus\{\bar x\},
\]
there is a local solution $X$ of \eqref{200215.2}  on $[0,\zeta)$, where 
	\[
		\zeta= \lim_{n \to \infty} \inf\left\{t\ge0\colon  {\rm dist}(X(t), \partial K \cup\{ \bar x\}) \leq \frac{1}{n} \right\}.
	\]
Since $X$ is a bounded martingale on $[0,\zeta)$, the martingale convergence theorem implies that $X(\zeta) = \lim_{t \to \zeta} X(t)$ exists. After time $\zeta$, we let $X$ continue like a scaled $\R^2$-valued Brownian motion.  The process $X$ is now an $\R^2$-valued martingale with $X(0) = x_0$  and  $\tr [X,X](t) = t$ for all $t \geq 0$ that satisfies the stochastic differential equation in \eqref{200215.2}  on $[0, \zeta]$. 
	
We now check the remaining properties. 	
	Since $\nablao w^\top \nabla w=0$ on $K^\circ$, It\^o's formula yields
\[
w(X(t)) = w(x_0) + \int_0^t \frac{\nablao w(X(s))^\top \nabla^2 w(X(s)) \nablao w(X(s))}{2|\nabla w(X(s))|^2} \d s, \quad t<\zeta.
\]
Since $w$ satisfies \eqref{eq:200214.4} on $K^\circ\setminus\{\bar x\}$ and is continuous on $K$, we get
\begin{align} \label{eq:coffee_store.1}
w(X(t)) = w(x_0) - t, \qquad t \leq \zeta.
\end{align}
Thus $w(X(\zeta)) \leq w(x_0) < w(\bar x)$, hence $\zeta$ is the first hitting time of the boundary $\partial K$ by $X$. Since $X$ follows a Brownian motion after $\zeta$, we actually have $\zeta = \tau_X$.  Evaluating \eqref{eq:coffee_store.1} at $t = \zeta = \tau_X$ yields $0 = w(x_0) - \tau_X$. 
 This proves the theorem for $x_0\ne\bar x$.

Suppose now that $x_0=\bar x$, and select a sequence $(x_n)_{n \in \N} \subset K$ with $\lim_{n \to \infty} x_n =  \bar x$. For each $n \in \N$, let $X^n$ be the $\R^2$-valued martingale constructed above with $X^n(0)=x_n$.  Fix $s \geq 0$ and $n \in \N$ and define 
\[
	M(t)=|X^n(t)-X^n(s)|^2-t  +s, \qquad t \geq s.
\]
Then $M$ is a local martingale, hence a supermartingale, on $[s, \infty)$ with $[M, M] \le 4\int_s^\cdot |X(u)-X(s)|^2 \d u$. Thus for all $t \geq s$,
\[
\E[[M,M](t)]\le 4\int_s^t \E[|X^n(u)-X^n(s)|^2] \d u \leq 4\int_s^t (u-s) \d u=2(t-s)^2,
\]
so that
\[
\E[|X^n(t)-X^n(s)|^4] = \E[(M(t)+t -s)^2] \le 2 \E[ [M, M](t)] + 2(t-s)^2 \le 6(t-s)^2.
\]
Thus by Kolmogorov's continuity criterion (see \citet{RY:99}, Theorem~I.2.1 and its proof) we get, for any fixed $\alpha\in(0,\frac14)$,
\[
\E\left[ \left( \sup_{0\le s<t \leq w(\bar x)} \frac{|X^n(t)-X^n(s)|}{|t-s|^\alpha} \right)^4 \right] \le c
\]
for some constant $c$ that does not depend on $n \in \N$. 

Since H\"older balls are relatively compact in $C([0,w(\bar x)], \R^2)$ by the Arzel\`a--Ascoli theorem, it follows that the sequence $(X^n|_{[0,w(\bar x)]})_{n\in\N}$ is tight  in $C([0,w(\bar x)], \R^2)$.  
Therefore, by Prokhorov's theorem there exists a continuous process $Y$ on $[0, w(\bar x)]$ such that a subsequence of $(X^n|_{[0,w(\bar x)]})_{n\in\N}$ converges weakly to $Y$ in $C([0,w(\bar x)], \R^2)$.
 Since $X^n(t) \in K$ and $w(X^n(t)) =w(x_n)-t$ for all $t\in[0,w(x_n)]$ and all $n\in\N$, the continuous mapping theorem and the continuity of $w$ give $Y(t) \in K$ and
\begin{equation}\label{eq_L_opt_sol_11}
	w(Y(t))=w(\bar x)-t
\end{equation}
for all $t \in [0, w(\bar x)]$.  We now define the process $X$ to be equal to $Y$ on $[0, w(\bar x)]$ and then continue like a scaled $\R^2$-valued Brownian motion.
This yields directly $w(\bar x) = \tau_X$. 
 
It now suffices to show that $X$ on $[0, w(\bar x)]$ is a weak solution to \eqref{200215.2}. By \citet[Proposition~5.4.6]{kar_shr_91} it is enough to argue that the process
\[
M^f(t) = f(X(t)) - \frac12 \int_0^t \frac{\nablao w(X(s))^\top  \nabla^2 f(X(s))  \nablao w(X(s))}{|\nabla w(X(s))|^2}\d s, \qquad t\in[0,w(\bar x)],
\]
is a martingale for every $f\in C^\infty(\R^2)$. The continuous mapping theorem implies that $M^f$ is a martingale whenever $f$ vanishes in a neighbourhood of $\bar x$. Now fix $s\in(0,w(\bar x)]$ and a general $f\in C^\infty(\R^2)$. Let $g\in C^\infty(\R^2)$ coincide with $f$ outside the set $\{x \in K\colon w(x)\geq w(\bar x)-s\}$ and vanish on $\{x \in K\colon w(x)\ge w(\bar x)-s/2\}$, which is a neighbourhood of $\bar x$. In particular, $M^g$ is a martingale. Moreover, \eqref{eq_L_opt_sol_11} implies that $X$ does not visit the set $\{x \in K\colon w(x)>w(\bar x)-s\}$ on $[s, w(\bar x)]$. As a result, $M^f(t)-M^f(s)=M^g(t)-M^g(s)$ for all $t \in [s,w(\bar x)]$, so that the process $(M^f(t))_{t\in[s,w(\bar x)]}$ is a martingale. Since $s\in(0,w(\bar x)]$ was arbitrary, this shows that  $(M^f(t))_{t\in(0,w(\bar x)]}$ is a martingale. The bounded convergence theorem finally implies that $(M^f(t))_{t\in[0,w(\bar x)]}$ is a martingale as well. 
This completes the proof. 
\end{proof}

\begin{remark}  \label{R:Pittsburgh2}
	We can again repeat the observation of Remark~\ref{R:Pittsburgh1}, now concerning Theorem~\ref{T:200215}. This result can be applied for any closed convex subset of $\R^2$, provided that we use the appropriate function $w$. 
In particular, assume,  for the moment only, that $K$ denotes the unit disk in $\R^2$.  Then it is easy to check that $w(x) = 1 - |x|^2$ for all $x \in K$ solves the corresponding boundary value problem in \eqref{200215.2}. The corresponding stochastic differential equation is \eqref{200214.1}, for which Theorem~\ref{T:200215} now yields the existence of a weak solution. 
\end{remark}

\subsection{Worst case time horizon} \label{SS:4.5}
Let us summarize where we stand. In the case $d = 3$, we have used the solution of the arrival time formulation of mean curvature flow to obtain an $\R^2$-valued martingale $X$ with $\tr [X, X](t) = t$ for all $t \geq 0$, whose first exit time from $K$ is deterministic (see Theorem~\ref{T:200215}). This yields the bound $v(x) \geq w(x)$ for all $x \in K$. We will now argue that this bound is optimal, which also shows that $X$ is optimal for the optimization problem on the right hand side of \eqref{eq:v}.

\begin{theorem}\label{T_MCF_repr}
Let $w$ be the classical solution of the boundary value problem~\eqref{eq:200214.4} from Theorem~\ref{T:200217.1}. Then $w = v$ on $K$ with $v$ given in \eqref{eq:v}; i.e., one has the stochastic representation
\[
w(x) = \sup_X \essinf \tau_X, \qquad x\in K,
\]
where the supremum extends over all $\R^2$-valued continuous martingales $X$ with $X(0)=x$ and $\tr [X, X](t)= t$ for all $t\ge0$, defined on arbitrary stochastic bases.
\end{theorem}

\begin{proof}
We use a verification argument. Theorem~\ref{T:200215}  exhibits, for every $x_0\in K$, a valid martingale $X$ with $X_0=x_0$ and $\essinf \tau_X = w(x_0)$. We now need to argue that this is the best one can do, so that this martingale is in fact optimal; i.e., it remains to show that $\essinf \tau_X \leq w(x_0)$ for any valid martingale $X$ with $X_0=x_0$. 

To this end, fix $x_0\in K$ and let $X$ be an $\R^2$-valued continuous martingale with $X(0)=x_0$ and $\tr [X, X](t) = t$ for all $t\ge0$, defined on a stochastic basis $(\Omega,\Fcal,(\Fcal_t)_{t\ge0},\P)$.  We need to  prove that $\essinf \tau_X\le w(x_0)$. Since the quadratic variation actually has absolutely continuous trajectories almost surely,  we have $[X, X] = \int_0^\cdot a(t) \d t$ for some predictable $\S^2_+$-valued process $a$ with $\tr(a) = 1$. Here $\S^2_+$ denotes the set of $2\times 2$ symmetric positive semidefinite matrices. 

Fix now $\lambda > 1$ and consider the function $w_\lambda: \lambda K \rightarrow [0, \infty),\, x \mapsto \lambda^2 w(x/\lambda)$. One easily checks that $w_\lambda$ satisfies the boundary value problem in \eqref{eq:200214.4} with $K$ replaced by $\lambda K$.  Thanks to the continuity of $w$ and since $\lambda>1$ was chosen arbitrarily it now suffices to argue that $\essinf \tau_X\le w_\lambda(x_0)$. We shall use repeatedly that $\sup_{x \in K} |\nabla w_\lambda(x)| < \infty$ (which follows from the fact that $\nabla w_\lambda$ is continuous on $\lambda K$).

Let now $k>0$ be a constant to be determined later, and define the process
\[
\widetilde X = X + k \int_0^{\cdot \wedge \tau_X} a(s)\nabla w_\lambda(X(s)) \d s.
\]
An application of It\^o's formula gives
\begin{equation}\label{eq_UB_1}
\begin{aligned}
t \leq t+w_\lambda(X(t)) = w_\lambda(x_0) &+ \int_0^t \nabla w_\lambda(X)^\top \d\widetilde X \\
&+ \int_0^t \Big( 1 + \frac12 \tr\left(a(s)\nabla^2 w_\lambda(X(s))\right) \\
&\qquad - k\nabla w_\lambda(X(s))^\top a(s)\nabla w_\lambda(X(s))\Big) \d s
\end{aligned}
\end{equation}
for all $t\le \tau_X$.

We now fix $\varepsilon \in (0,1)$ and claim that there exists $k>0$ such that the integrand in the $\d s$-integral on the right hand side of \eqref{eq_UB_1} is bounded from above by $\varepsilon$. To this end it suffices to show the existence of $k>0$ such that 
\begin{align} \label{eq:200226.2}
	1 + \frac12 \tr\left(A \nabla^2 w_\lambda(x)\right) \leq \varepsilon +  k\nabla w_\lambda(x)^\top A\nabla w_\lambda(x)
\end{align}
for all $(x,A) \in K \times S^2_+$ with $\tr(A) = 1$.  Let us argue by contradiction and assume that no such $k$ exists.  Then, since $\nabla^2 w_\lambda$ is bounded on $K$, there exists a sequence $(x_n,A_n)_{n \in \N} \subset K \times\S^2_+$ such that $\tr(A_n)=1$,  $\lim_{n \to \infty} \nabla w_\lambda(x_n)^\top A_n \nabla w_\lambda(x_n)= 0$, and $1+(1/2)\tr(A_n \nabla^2 w_\lambda(x_n)) > \varepsilon$. After passing to a subsequence, we have $\lim_{n \to \infty}  x_n= \widetilde x$ and  $\lim_{n \to \infty} A_n = A$ for some $(\widetilde x, A) \in K \times \S^2_+$ with  
\begin{align} \label{eq:200226.1}
	\tr(A) = 1, \qquad  \nabla w_\lambda(\widetilde x)^\top A \nabla w_\lambda(\widetilde x) = 0, \qquad \text{and} \qquad 1+\frac{1}{2} \tr(A\nabla^2 w_\lambda(\widetilde x)) \geq \varepsilon.
\end{align} 
Let us first argue that $\nabla w_\lambda(\widetilde x) \neq 0$. If $\nabla w_\lambda(\widetilde x) = 0$ then $\nabla^2 w_\lambda(\widetilde x) = -2 I$ by Theorem~\ref{T:200217.1} (recall also Remark~\ref{R:Pittsburgh1}), contradicting the first and last equality in \eqref{eq:200226.1}. Hence we have  $|\nabla w_\lambda(\widetilde x)| \neq 0$. This forces $\rk(A)=1$ and then $A = \nablao w_\lambda(\widetilde x) \nablao w_\lambda(\widetilde x)^\top / |\nabla w_\lambda(\widetilde x)|^2$. Thus, the last equality of 
\eqref{eq:200226.1} becomes
\[
1+\frac{\nablao w_\lambda(\widetilde x)^\top \nabla^2 w_\lambda(\widetilde x) \nablao w_\lambda(\widetilde x)}{ |\nabla w_\lambda(\widetilde x)|^2}  \geq \varepsilon.
\]
This contradicts the fact that $w_\lambda$ solves the boundary value problem in \eqref{eq:200214.4} (with $K$ replaced by $K_\lambda$).

We now argue that $(1-\varepsilon) \essinf \tau_X\le w_\lambda(x_0)$. Since $\varepsilon>0$ was arbitrary, this will complete the proof. 
With $k$ as determined in the previous paragraph, \eqref{eq_UB_1} and \eqref{eq:200226.2} yield
\[
	(1-\varepsilon) t - w_\lambda(x_0)  \leq \int_0^t \nabla w_\lambda(X)^\top \d\widetilde X, \qquad t \leq \tau_X.
\]
Consider now the strictly positive stochastic exponential $Z$, given by the stochastic differential equation
\[
	Z = 1 - k \int_0^{\cdot \wedge \tau_X} Z \nabla w_\lambda(X)^\top \d X.
\]
Since $\essinf \tau_X < \infty$ and $\nabla w_\lambda$ is bounded on $K$, Novikov's condition shows that $Z$ is a martingale on $[0,\essinf \tau_X]$. Let $\Qu$ be the probability measure on $\Fcal_{\essinf \tau_X}$ induced by $Z_{\essinf \tau_X}$. Under $\Qu$, $\int_0^{\cdot}\nabla w_\lambda(X)^\top \d\widetilde X$ is a local martingale on $[0,\essinf \tau_X]$ bounded from below, hence a supermartingale. Consequently,
\[
(1-\varepsilon) \essinf \tau_X - w_\lambda(x_0) 
\leq \E_\Qu\left[ \int_0^{\essinf \tau_X} \nabla w_\lambda(X)^\top \d\widetilde X \right] \leq 0,
\]
yielding the claim. This completes the proof.
\end{proof}

We have now argued that $v(x) = w(x)$ for all $x \in K$, where $v$ is given in \eqref{eq:v} and $w$ denotes the arrival time function of the mean curvature flow provided by Theorem~\ref{T:200217.1}. Thanks to the observations in Subsection~\ref{SS:4.1} we thus have $T_* = \sup_{x \in K} w(x)$. While an explicit expression for $w$ is not available, its maximal value can be computed using \citet[Lemma~3.1.7]{GH:86}. It is observed there that the area $A$ enclosed by a smooth simple closed curve that flows by mean curvature satisfies
\[
\frac{\d A(t)}{\d t} = -2\pi.
\]
The extinction time of the flow is therefore $A(0)/2\pi$, which is the time it takes until the area becomes zero. In our case, the initial contour is not smooth. However, the arguments in the proof of Theorem~\ref{T:200217.1} show that it immediately becomes smooth under the mean curvature flow, and it follows that the formula for the extinction time is still valid. Since $w$ is twice the arrival time of the mean curvature flow we obtain
\[
w(0) = \max_{x\in K}w(x) = 2 \frac{{\rm area}(K)}{2\pi} =  \frac{{\rm area}(\Delta^3)}{\pi} = \frac{\sqrt{3}}{2\pi} \approx 0.28,
\]
where the first equality holds due to Remark~\ref{R:Pittsburgh1}, and the two areas coincide due to the isometry property of $U\colon\Delta^d\to K$ (see Subsection~\ref{SS:4.1}). We arrive at the following result.

\begin{theorem} \label{T*}
If $d=3$, then the smallest time horizon beyond which any sufficiently volatile market admits relative arbitrage is
\[
T_* = \frac{\sqrt{3}}{2\pi} \approx 0.28.
\]
\end{theorem}

Note in particular that the true value of $T_*$ lies strictly between the previously best known lower and upper bounds $1/6$ and $2/3$; see the end of Section~\ref{S:2}.

 \section{The general case $d \ge 3$ and minimum curvature flow} \label{S:5}

We now turn to the general case $d\ge3$. The representation of $T_*$ in Theorem~\ref{T_Tstar_repr} is still valid, and we approach it via the value function $u$ in \eqref{eq:u}. Again, $u$ can be characterized as the solution of a partial differential equation, but one that is more complicated than in the case $d=3$. There are two essential differences.

First, it turns out that $u$ no longer vanishes on the (relative) boundary of $\Delta^d$. To see why, imagine a market $\mu$ that hits the interior of a $(d-1)$-dimensional boundary face of $\Delta^d$. If $d \geq 4$, the construction in Subsection~\ref{SS:4.2} shows that an optimal market will not immediately exit $\Delta^d$, but instead spend some deterministic amount of time on this boundary face. This affects how we deal with boundary conditions.

Second, the equation itself no longer describes the arrival time of the mean curvature flow. Instead, it corresponds to another flow that we call the \emph{minimum curvature flow}, which is more degenerate than the mean curvature flow. This flow is closely related to the \emph{codimension-$(d-1)$ mean curvature flow} of \cite{amb_son_96}.

As in Subsection~\ref{SS:4.1}, we use an affine isometry $U\colon\R^d\to\R^{d-1}$ to identify $\Delta^d$ with a polytope $K\subset\R^{d-1}$ with nonempty interior $K^\circ\ne\emptyset$. Then $u(y)=v(Uy)$ for all $y\in\Delta^d$, where
\begin{equation}\label{eq_val_fcn_d}
v(x) = \sup\left\{ \essinf \tau_X\colon \ \ \begin{minipage}[c][2em][c]{.55\textwidth}
$X$ is an $\R^{d-1}$-valued  continuous martingale with\\  $X(0) = x$ and $\tr[X, X](t)= t$ for all $t \geq 0$
\end{minipage}
\right\};
\end{equation}
here $\tau_X$ denotes the first exit time of $X$ from $K$. Let us now describe heuristically how to obtain the partial differential equation satisfied by $v$.

Consider an $\R^{d-1}$-valued continuous martingale $X$ with $\tr[X,X](t)=t$ for all $t\ge0$. Its quadratic variation can be written
\[
[X,X] = \int_0^\cdot a(s) \d s
\]
for some $\S^{d-1}_+$-valued process $(a(t))_{t\ge0}$ with $\tr a(t)=1$ for all $t\ge0$. For the sake of discussion, suppose $v$ is $C^2$. Suppose also that $v$ is zero on $\partial K$, so that $v(X(\tau_X))=0$. (In reality this only holds for strictly convex domains. Nevertheless, it still produces correct heuristics, basically because $v(X(\tau_X))=0$ still holds for the optimal choice of $X$.) It\^o's formula then yields
\[
\tau_X = v(X(0)) + \int_0^{\tau_X} \nabla v(X)^\top \d X + \int_0^{\tau_X}\left( 1 + \frac12\tr\left(a(t)\nabla^2v(X(t))\right)\right) \d t.
\]
We look for choices of $X$ that lead to deterministic lower bounds on $\tau_X$. To do so, we focus on those $X$ for which the stochastic integral above vanishes, i.e., $\nabla v(X(t))^\top \d X(t)=0$ for all $t \in [0, \tau_X]$. That is, we require that $a\nabla v(X)=0$. We then obtain
\begin{align*}
\tau_X &= v(X(0)) + \int_0^{\tau_X}\left( 1 + \frac12\tr\left(a(t)\nabla^2v(X(t))\right)\right) \d t \\
&\le v(X(0)) + \int_0^{\tau_X} \left( 1 - F\left(\nabla v(X(t)),\nabla^2 v(X(t))\right) \right) \d t,
\end{align*}
where for any $p\in\R^{d-1}$ and $M\in\S^{d-1}$ we define
\begin{equation}\label{eq_FpM}
F(p,M) = \inf\left\{-\frac12\tr(aM)\colon a\in\S^{d-1}_+,\ \tr(a)=1,\ ap=0\right\}.
\end{equation}
This suggests that $v$ should satisfy the partial differential equation
\begin{equation}\label{eq_FpM_PDE}
F(\nabla u,\nabla^2 u)=1.
\end{equation}
Indeed, in this case we obtain $\essinf\tau_X\le \tau_X\le v(X(0))$, and if additionally $a(t)$ achieves the infimum in \eqref{eq_FpM} with $p=\nabla v(X(t))$ and $M=\nabla^2 v(X(t))$, we obtain $\tau_X=v(X(0))$. This heuristic reasoning makes the following theorem plausible. The rigorous proof is more involved, and can be found in the companion paper \cite{LR:2020a}.

\begin{theorem}
For $d\ge3$, the value function $v$ in \eqref{eq_val_fcn_d} is the unique continuous (on $K$) viscosity solution of \eqref{eq_FpM_PDE} in $K^\circ$ with zero boundary condition (in the viscosity sense). Moreover, $v$ is quasi-concave, vanishes on the vertices and boundary lines of $K$, and is strictly positive elsewhere in $K$.
\end{theorem}

\begin{proof}
Existence and uniqueness in the class of upper semicontinuous viscosity solutions follows from \citet[Theorem~1.2]{LR:2020a}. Continuity on $K$ follows from \citet[Theorem~1.4 and Corollary~5.9]{LR:2020a}. The remaining statements follow from \citet[Theorem~1.3]{LR:2020a}.
\end{proof}

We will not delve into the technical aspects of this theorem here, nor its connection to the minimum curvature flow; details are available in \citet{LR:2020a}. For completeness we nonetheless give the relevant definition of viscosity solution. A  bounded function $u\colon K \to\R$ is called a viscosity subsolution of \eqref{eq_FpM_PDE} in $K^\circ$ if
\[
\left.
\begin{minipage}[c][3em]{.35\textwidth}\center
$(\bar x,\varphi)\in K^\circ \times C^2(\R^{d-1})$ and \\[1ex] $(u^*-\varphi)(\bar x) = \max_K(u^*-\varphi)$
\end{minipage}
\right\}
\quad\Longrightarrow\quad\text{$F_*(\nabla\varphi(\bar x),\nabla^2\varphi(\bar x))\le1$,}
\]
where an upper (lower) star denotes upper (lower) semicontinuous envelope. We say that $u$ has zero boundary condition (in the viscosity sense) if
\[
\left.
\begin{minipage}[c][3em]{.35\textwidth}\center
$(\bar x,\varphi)\in \partial K \times C^2(\R^{d-1})$ and \\[1ex] $(u^*-\varphi)(\bar x) = \max_K(u^*-\varphi)$
\end{minipage}
\right\}
\quad\Longrightarrow\quad\text{$F_*(\nabla\varphi(\bar x),\nabla^2\varphi(\bar x))\le1$ or $u^*(\bar x)\le 0$.}
\]
The function $u$ is said to be a viscosity supersolution in $K^\circ$ with zero boundary condition if the above properties hold with $u^*$, $F_*$, $\max$, $\le$ replaced by $u_*$, $F^*$, $\min$, $\ge$. It is a viscosity solution in $K^\circ$ with zero boundary condition if it is both a viscosity sub- and supersolution in $K^\circ$ with zero boundary condition.

\begin{remark}
Consider the case $d = 3$. For any $p\in\R^2\setminus\{0\}$, if $a\in\S^2_+$ satisfies $ap=0$, then $\rk(a)\le1$ and we must have $a=qq^\top$ for some $q\in\R^2$ orthogonal to $p$. If in addition $\tr(a)=1$ then $|q|=1$. Thus, in view of \eqref{eq_FpM}, provided $\nabla u \neq 0$, we obtain
\[
F(\nabla u,\nabla^2 u) = -\frac{\nablao u^\top\nabla^2 u \nablao u}{2|\nabla u|^2},
\]
and \eqref{eq_FpM_PDE} reduces to the partial differential equation in \eqref{eq:200214.4}.
\end{remark}

\begin{remark}
Recall that $K = U(\Delta^d)$ for an affine isometry $\R^d \rightarrow \R^{d-1}$. Let us assume that $U$ maps $(1/d) \bf{1}$ to the origin, so that $K$ is $(d-1)$-dimensional polytope centered at the origin.
It is worth noting that the function $Q(x)=\frac{d-1}{d}-|x|^2$ is nonnegative on $K$ and is the solution of the partial differential equation 
\[
\inf\left\{ - \frac12\tr(a\nabla^2 u) \colon a\in \S^{d-1}_+,\ \tr(a)=1\right\}=1
\]
in $K^\circ$ with zero boundary condition (in the viscosity sense). This equation is similar but not identical to \eqref{eq_FpM_PDE}. It can be shown that it is the dynamic programming equation corresponding to the control problem 
\[
\sup\left\{ \E[\tau_X] \colon \ \ \begin{minipage}[c][2em][c]{.55\textwidth}
$X$ is an $\R^{d-1}$-valued  continuous martingale with\\  $X(0) = x$ and $\tr[X, X](t)= t$ for all $t \geq 0$
\end{minipage}
\right\},
\]
where the expected exit time is maximized, rather than the essential infimum.
\end{remark}

 \section{Conclusion and open problems} \label{S:6}

Let us summarize our findings. For sufficiently volatile markets with $d=3$ assets, we have shown that relative arbitrage always exists beyond the critical time horizon $T_*=\frac{\sqrt{3}}{2\pi}$, but not always before $T_*$. The value of $T_*$ is determined by analyzing the arrival time function of mean curvature flow in the plane. In general for $d\ge3$, the critical time horizon $T_*$ is determined through a partial differential equation connected to minimum curvature flow, which happens to coincide with the mean curvature flow for $d=3$.

We conclude with four open problems.
\begin{itemize}
\item 
We have seen that in any sufficiently volatile market with $d=3$ assets, the portfolio generating function $Q(x)=1-|x|^2$ generates relative arbitrage over $[0,T]$ for any $T>\frac23$. However, relative arbitrage exists already for $T\in (\frac{\sqrt{3}}{2\pi},\frac23]$ (see Figure~\ref{F_intermediate_regime}). What does the arbitrage strategy look like for these intermediate time horizons? Is there a portfolio generating function that works in every sufficiently volatile market? More generally, does there exist a single path-dependent strategy $\theta= (\theta(t,(\mu(s))_{s\le t}))_{t \in [0,T]}$ that generates arbitrage over the time horizon $[0,T]$ with $T\in (\frac{\sqrt{3}}{2\pi},\frac23]$, in any sufficiently volatile market? Or is the arbitrage strategy inherently model dependent? We remark that the arrival time function $w$ of the mean curvature flow will not yield such a portfolio generating function. This is because it fails to be concave near the vertices of the triangle. 
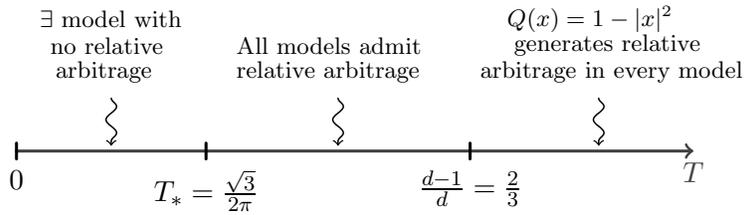
\begin{figure}[h]
\begin{center}
\begin{tikzpicture}
  [line cap=round,line join=round,x=4.5cm,y=4cm,decoration={brace,amplitude=2pt},set/.style={draw}]

  \draw[color=darkgray,line width=0.4mm,->] (0,0) -- (2,0) node[below] {$T$};
  \draw[line width=0.4mm] (0,3pt) -- (0,-3pt) node[below] {$0$};
  \draw[line width=0.4mm] (2*0.28,3pt) -- (2*0.28,-3pt) node[below] {$T_*=\frac{\sqrt{3}}{2\pi}$};
  \draw[line width=0.4mm] (2*0.67,3pt) -- (2*0.67,-3pt) node[below] {$\frac{d-1}{d}=\frac23$};

  \draw[color=black] (5pt,50pt) node[right] {\footnotesize $\exists$ model with};
  \draw[color=black] (9pt,40pt) node[right] {\footnotesize no relative};
  \draw[color=black] (11pt,30pt) node[right] {\footnotesize arbitrage};

  \draw[color=black] (2*0.28+.06,40pt) node[right] {\footnotesize All models admit};
  \draw[color=black] (2*0.28+.06,30pt) node[right] {\footnotesize relative arbitrage};

  \draw[color=black] (2*0.67+0.08,50pt) node[right] {\footnotesize $Q(x)=1-|x|^2$};
  \draw[color=black] (2*0.67+0.1,40pt) node[right] {\footnotesize generates relative};
  \draw[color=black] (2*0.67,30pt) node[right] {\footnotesize arbitrage in every model};

  \draw[line width=0.2mm,->,snake it] (0.28,20pt) -- (0.28,2pt);
  \draw[line width=0.2mm,->,snake it] (0.28+0.67,20pt) -- (0.28+0.67,2pt);
  \draw[line width=0.2mm,->,snake it] (0.28+2*0.67+0.1,20pt) -- (0.28+2*0.67+0.1,2pt);
\end{tikzpicture}
\end{center}
\caption{The various regimes for existence of relative arbitrage in the case $d=3$.}\label{F_intermediate_regime}
\end{figure}

\item We have not provided any numerical values for $T_*$ in dimension $d>3$. Can one use numerics to produce quantitative bounds? In particular, lower (upper) bounds can be obtained by searching for explicit subsolutions (supersolutions) of \eqref{eq_FpM_PDE}. 

\item In Definition~\ref{D_suff_vol} we defined a sufficiently volatile market in terms of $\tr[\mu,\mu]$. This is neither the only possible definition, nor the most natural. It can be generalized as follows.
\begin{definition} \label{D:200620}
Given a portfolio generating function $G$, a market $\mu$ is called {\em $G$-sufficiently volatile} if $\Gamma^G$ in \eqref{eq:GammaG} is well defined and satisfies $\Gamma^G(t)\ge t$ for all $t\ge0$.
\end{definition}
The condition \eqref{eq:200212.2} arises by taking $G(x)=Q(x)=1-|x|^2$. But there are other choices. For instance, one could take the \emph{entropy function} $G_1(x)= -\sum_{i = 1}^d x_i \log x_i$, or the \emph{geometric average} $G_2(x)= \prod_{i=1}^d x_i^{1/d}$. For a given choice of $G$, one can then ask: what is the smallest time horizon beyond which relative arbitrage is possible in any $G$-sufficiently volatile market? If $G = G_2$, then $G|_{\partial \Delta^d} = 0$ and the answer seems to be that the critical time horizon is $\max_{y \in \Delta^d} G_2(y)= 1/d$; see also \citet[Remark~6.18]{Fernholz:Karatzas:Ruf:2016}. If $G = G_1$, then $G$-sufficiently volatile markets are those that satisfy \eqref{eq:200212.1}. It remains an open problem how to determine the critical time horizon for relative arbitrage across $G_1$--sufficiently volatile markets (or other choices of $G$ other than $Q$ and $G_2$). We conjecture that the arguments in this paper can be adapted to answer this open problem, and might lead to new geometric flows. In this context, let us also mention \cite{Pal:Wong:18}, who consider the geometries corresponding to different portfolio generating functions.
\item Expanding on the previous question, one might ask how to use additional statistical knowledge about the market to further bound the smallest time beyond which relative arbitrage is possible in any such market. For example, assume $d = 3$ and restrict focus to sufficiently volatile markets that are in addition diverse, i.e., $\max_{i=1,2,3} \mu_i \leq 1-\delta$ for some $\delta \in (0, 1/2)$. What is the minimum time horizon now? Clearly, it is smaller than $\sqrt{3}/(2\pi)$.
For this setup, when $d=3$ and the markets are diverse, the answer is easy thanks to Remark~\ref{R:Pittsburgh1} and the observations before Theorem~\ref{T*}. One only needs to compute the area of the region $\widetilde K=\{x\in\Delta^3\colon \max_i x_i \le 1-\delta\}$, which is $\sqrt{3}(1-3\delta^2)/2$. This yields the minimum time
\[
	\frac{1}{\pi} \text{area}\big(\widetilde K\big) = \frac{\sqrt{3}}{2 \pi} \left(1 - 3 \delta^2\right).
\]
\end{itemize}

% ----------------------------------------------------------------
\bibliographystyle{plainnat}
\bibliography{bibl, aa_bib}

\end{document}